%% file: balancing.tex
\documentclass[letterpaper, 10 pt, conference]{ieeeconf}  

\IEEEoverridecommandlockouts                              
\usepackage{amsfonts,amsmath,amssymb}
\usepackage{graphicx}

\usepackage{ifthen}
\usepackage{color}
\usepackage{cite}

\usepackage{amsfonts}
\usepackage{amsmath}
\usepackage{amssymb}
\usepackage{color}
\usepackage{graphics}
\usepackage{algorithmicx}
\usepackage{algpseudocode}
\usepackage[boxruled,vlined,linesnumbered,onelanguage]{algorithm2e}
\usepackage{caption}
\usepackage{epsfig}
\usepackage{chngcntr}
\usepackage{dashrule}
\usepackage[position=b]{subcaption}
\usepackage[hyperfootnotes=false]{hyperref}
\usepackage{listings}
\usepackage{color}
\usepackage{authblk}

\usepackage[multiple]{footmisc}

\usepackage{xspace}
\usepackage{tikz}
\usepackage{subcaption}
\usetikzlibrary{calc, positioning, fit, arrows, arrows.meta, shapes.geometric, shapes.symbols, shapes.misc, patterns}

\newcommand{\R}{{\rm I\!R}}

\def\scr#1{{\cal #1}}
\newtheorem{theorem}{Theorem}
\newtheorem{lemma}{Lemma}
\newtheorem{proposition}{Proposition}

\newtheorem{example}{Example}

\def\qed{ \rule{.08in}{.08in}}

\input{steve.tex}

\newcommand{\drawnodex}[6]{
	\path (#1) node[draw,thick,circle,minimum width=1em,minimum height=1em,inner sep=0.2ex] (#2) {} node[inner sep=0.2ex,fill=white,#3=0 of #2] {#4} node[inner sep=0.2ex,fill=white,#5=0 of #2] {#6};
}

\newcommand{\drawfoe}[2]{
	\path (#1) -- (#2) node(_1)[allow upside down,sloped,above,minimum height=1.2em,anchor=center,pos=0.3] {};
	\path (#1) -- (#2) node(_2)[allow upside down,sloped,above,minimum height=1.2em,anchor=center,pos=0.7] {};
	\draw[-stealth,ultra thick,red] (#1) .. controls (_1.north) and (_2.north) .. (#2);
	\draw[-stealth,ultra thick,red] (#2) .. controls (_2.south) and (_1.south) .. (#1);
}

\newcommand{\drawfoex}[5][0.5]{
	\path (#2) -- (#3) node(_1)[allow upside down,sloped,above,minimum height=1.2em,anchor=center,pos=0.3] {};
	\path (#2) -- (#3) node(_2)[allow upside down,sloped,above,minimum height=1.2em,anchor=center,pos=0.7] {};
	\draw[-stealth,ultra thick,red] (#2) .. controls (_1.north) and (_2.north) .. (#3) node[inner sep=0.2ex,fill=white,sloped,pos=#1] {#4};
	\draw[-stealth,ultra thick,red] (#3) .. controls (_2.south) and (_1.south) .. (#2) node[inner sep=0.2ex,fill=white,sloped,pos=#1] {#5};
}

\begin{document}

\title{The Power Allocation Game on A Network: Balanced Equilibrium}


\author{{\rm Yuke Li, and A. Stephen Morse} \thanks{
This work was supported by National Science Foundation grant n.1607101.00, and US Air Force grant n. FA9550-16-1-0290. Y.~Li and A. S.~Morse are respectively with the Department of Political Science and the Department of Electrical Engineering, Yale University (\texttt{\{yuke.li, as.morse\}@yale.edu}). 
}}

\maketitle

\begin{abstract}

This paper studies a special kind of equilibrium termed as ``balanced equilibrium'' which arises in the power allocation game defined in \cite{allocation}. In equilibrium, each country  in antagonism has to use all of its own power to counteract received threats,
and the ``threats'' made to each adversary just balance out the threats received from that adversary.  This paper establishes conditions on different types of networked international environments in order for this equilibrium to exist. The paper also links the existence of this type of equilibrium    on structurally balanced graphs to the Hall's Maximum Matching problem and the Max Flow problem. 


\end{abstract}

\section{Introduction}



In a recent paper \cite{allocation}, a power allocation game  on networks  is developed to study countries' strategic behaviors for allocating resources among one another in an international environment.  In playing the game, each  country allocates its  total power among its   friends and foes in order to ensure the survival of its friends and itself while opposing the survival  of  its foes.

This paper investigates the existence of a  `balanced equilibrium' in a power allocation game. The motivation for this comes from both the simple technical structure of the equilibrium of interest and its real world implications. A condition  for a balanced equilibrium to exist in the  power allocation game  is  similar in  form to the condition for  existence of an equilibrium  as  expressed in  Hall's maximum matching theorem (see \cite{ford2015flows,lovasz2009matching,gerards1995matching,hall1935representatives}). 

This paper belongs to a vast literature on agents' interactions on signed graphs. This literature has its roots in the study of consensus in distributed, multi-agent systems, where the earliest work include \cite{tsitsiklis1986distributed, jadbabaie2003coordination, olfati2007consensus}. The literature 
 has recently been broadened to  address  scenarios where antagonistic interactions also exist; correspondingly, signed graphs become a focus in the analysis (i.e., a signed graph is a graph where each edge has a positive sign denoting a cooperative relation or a negative sign denoting a conflictual relation)  \cite{altafini2015predictable, altafini2013consensus, altafini2012dynamics, proskurnikov2016opinion,liu2016discrete,shi2016evolution,liu2015stability,proskurnikov2014consensus,valcher2014consensus, meng2016behaviors}.  A particular type of signed graph called a  `structurally balanced graph' \cite{cartwright1956structural}) plays a key role in the `modulus consensus'  problem  (e.g., \cite{altafini2013consensus, valcher2014consensus}).

In contrast with the consensus problem on signed graphs, an optimization framework is necessarily needed for studying agents' strategic interactions on signed graphs. Moreover,  a game-theoretic framework can naturally be formulated to capture the noncooperative and cooperative scenarios innate to those interactions. It appears that \cite{konig2017networks},  \cite{allocation},  \cite{survival} and this paper are the first set of papers on the topic of  `games on signed graphs'; a difference between the power allocation game in the latter three papers and \cite{konig2017networks} is that each agent has a total resource constraint in the power allocation game. It should also be noted that ``games on signed graphs'' means  games where each agent's strategy does not involve changing their network of relations. Games where agents' strategies do involve changing networks are usually referred to as  `network formation games' \cite{jackson2005survey}; network formation games on countries' relation dynamics include \cite{hiller2017friends}, \cite{jackson2015networks} and \cite{relation}.  In particular, \cite{hiller2017friends} studies what kinds of networks in Nash equilibrium are structurally balanced).

The organization of this paper is as follows. First, in Section \ref{pag} the power allocation game is briefly summarized. Second, in Section \ref{bal} the notion of a  balanced equilibrium will be defined and its real world implications will be discussed. Third, the existence of such an equilibrium will be investigated in power allocation games  for different networked environments. In particular,  a connection will be drawn between balanced equilibria  and the Hall Maximum Matching problem and the Max Flow problem in combinatorial optimization.

\section{The Power Allocation Game} \label{pag}

\subsection{Formulation} By the  {\em power allocation game}  is meant a distributed resource allocation game between $n$ countries with labels in $\mathbf{n}= \{1,2,\ldots,n\}$\cite{allocation}. The game is formulated on a simple, undirected,  signed graph $\mathbb{E} = \{\mathcal{V}, \mathcal{E}\}$ called ``an environment graph'' \cite{survival} whose $n$ vertices correspond to the countries and whose $m$ edges represent relationships between countries. An edge between distinct vertices $i$ and $j$, denoted by $(i,j)$, is labeled with a plus sign if countries $i$ and $j$ are friends and with a minus sign if countries $i$ and $j$ are adversaries. Let the set of all friendly pairs be $\mathcal{R}_{\mathcal{F}}$ and the set of all adversarial pairs be $\mathcal{R}_{\mathcal{A}}$. For each $i\in\mathbf{n}$, $\scr{F}_i$ and $\scr{A}_i$ denote the sets of labels of country $i$'s friends and adversaries respectively; it is assumed that $i\in\scr{F}_i$ and that $\scr{F}_i$ and $\scr{A}_i$ are disjoint sets.


Each country $i$ possesses a nonnegative quantity $p_i$ called the {\em total power} of country $i$. An allocation of this power,  called a {\em strategy}, is a nonnegative $1\times n$ row vector $u_i$ whose $j$th component $u_{ij}$ is that part of $p_i$ which country $i$ allocates under the strategy to either support country $j$ if $j\in\scr{F}_i$ or to the demise of country $j$ if $j\in\scr{A}_i$; accordingly $u_{ij}= 0$ if $j\not\in\scr{F}_i\cup\scr{A}_i$ and $u_{i1}+u_{i2}+\cdots +u_{in} = p_i$. The goal of the game is for each country to choose a strategy which contributes to the demise of all of its adversaries and to the support of all of its friends.

Each set of country strategies $\{u_i,\;i\in\mathbf{n}\}$ determines an
 $n\times n$ matrix  $U$ whose $i$th row is $u_i$. 
 Thus $U = \begin{bmatrix}u_{ij}\end{bmatrix}_{n\times n}$ is a nonnegative matrix such that, for each $i\in\mathbf{n}$, $u_{i1}+u_{i2}+\cdots +u_{in} = p_i$. Any such matrix is called a {\em power allocation matrix}  and $\scr{U}$ is the set of all $n\times n$ strategy matrices.

\subsection{Multi-front Pursuit of Survival}

Just how  each country  allocates its  power in a power allocation game \cite{allocation,survival} is consistent   with the fundamental assumptions about countries' behaviors in classical international relations theory \cite{waltz2010theory}.
In particular, 
each power allocation matrix $U$ determines for each $i\in\mathbf{n}$, the {\em total support} $\sigma_i(U)$ of
 country $i$ and the {\em total threat}  $\tau_i(U)$ against country $i$.  Here
 $\sigma_i:\scr{U}\rightarrow\R$ and $\tau_i:\scr{U}\rightarrow \R$ are non-negative
 valued maps defined by
$U\longmapsto \sum_{j\in\scr{F}_i}u_{ji} +\sum_{j\in\scr{A}_i}u_{ij}$
 and $U\longmapsto \sum_{j\in\scr{A}_i}u_{ji}$
  respectively. Thus country $i$'s total support is the sum of the amounts of power
   each of country $i$'s friends
  allocate to its support plus the sum of the amounts of power country $i$ allocates  to the
  destruction of  all of its adversaries. Country $i$'s total threat, on the other hand, is the sum of the amounts of power
  country $i$'s adversaries allocate to its destruction. These allocations in turn determine
   country $i$'s {\em state} $x_i(U)$ which may be safe, precarious, or unsafe depending on the relative
    values of $\sigma _i(U)$ and $\tau_i(U)$. In particular, 
     $x_i(U) = $ safe
    if $\sigma_i(U)>\tau_i(U)$, $x_i(U)=$ precarious if $\sigma_i(U)=\tau_i(U)$, or
    $x_i(U) = $ unsafe if $\sigma_i(U)<\tau_i(U)$.

In playing the power allocation game, each country  selects its own  strategy  in accordance with certain weak and/or strong preferences.
A sufficient set of conditions for  country $i$ to {\em weakly prefer} power allocation matrix $V\in\scr{U}$ over power allocation matrix $U\in\scr{U}$ are as follows
\begin{enumerate}
\item  For all $j\in\scr{F}_i$
 either $x_j(V)\in$ \{safe, precarious\}, or $x_j(U)\in$ \{unsafe\}, or both.
 \item For all $j\in\scr{A}_i$ either $x_j(V)\in $ \{unsafe, precarious\}, or $x_j(U)\in$ \{safe\}, or both.
\end{enumerate}
Weak preference  by country $i$ of  $V$ over $U$ is denoted by $U \preceq V$.
Meanwhile, a sufficient condition for country
 $i$ to be {\em indifferent} to the choice between $V$ and $U$ is  that $x_i(U)=x_j(V)$ for all $j\in\scr{F}_i\cup\scr{A}_i$.  This is denoted by $V\sim U$.
Finally, a sufficient condition for country $i$  to {\em strongly prefer} $V$ over $U$ is that $x_i(V)$ be a  safe or precarious  state and $x_i(U)$ be an unsafe state. Strong preference by country $i$ of $V$ over $U$ is denoted by $U \prec V$.

 The Nash equilibrium concept is naturally employed in the power allocation game  to make predictions. Let country $i$'s deviation from the power allocation matrix $U$ be a nonnegative-valued $1\times n$ row vector $d_{i}$ such that $u_{i} + d_{i}$ is a valid strategy that satisfies the total power constraint for country $i$. The deviation set $\mathcal{D}_{i}(U)$ is the set of all possible deviations of country $i$'s power from the power allocation matrix $U$. In the context of the power allocation game, a power allocation matrix $U$ is a pure strategy Nash Equilibrium if no one deviation in strategy by any single country $i$ is `profitable' for country  $i$. In other words 
\begin{equation*}
U + e_{i} d_{i} \preceq U, \;\;\;\;\; {\rm for\; all} \;\;\;\;\; d_{i} \in \mathcal{D}_{i}(U),
\end{equation*} 
where $e_{i}$ is the  $i$th  unit $n$-vector.

Denote by $\mathcal{U}^{*}$ the set of pure strategy Nash equilibria. Call 
 $U\in\mathcal{U}^*$   {\em equilibrium equivalent} to $V\in\mathcal{U}^*$ if and only if $x(U) = x(V)$. 
 Clearly the relation `equilibrium equivalence' is the equivalence kernel of the restriction of $x$ to $\scr{U}^*$ and thus is an equivalence relation on $\mathcal{U}^{*}$.
Let $[U]_*$ be the \emph{equilibrium equivalence class} of $U \in \mathcal{U}^{*}$. Obviously, the total number of equilibrium equivalence classes is at most $3^{n}$, which in turn is the cardinality of the co-domain of $x$.

 
\section{Balanced Equilibrium} \label{bal}
The aim of this section is to explain what is meant by a balanced equilibrium. We begin with a motivating example.

\noindent{\bf Example 1:} Consider an environment consisting of $n=3$ countries whose total powers are $p_1 = 8$, $p_2 = 6$ and $p_3 = 4$.  Assume that all three are adversaries of each other in which case $\mathbb{E} $ is a complete graph with $-$ assigned to all three edges. There are at least  four possible equilibrium equivalence classes $[U_i]_*, i\in\{1,2,3,4\}$ where
$$U_1 = \begin{bmatrix} 2  &4 & 2\\ 2 & 0 & 4\\ 0 & 4 & 0  \end{bmatrix},\hspace{.1in}
U_2 = \begin{bmatrix} 0  &4 & 4\\ 5 & 0 & 1\\ 4 & 0 & 0  \end{bmatrix},
\hspace{.1in} U_3 = \begin{bmatrix} 0  &6 & 2\\ 6 & 0 & 0\\ 3 & 1 & 0  \end{bmatrix},$$ and
$$U_4 = \begin{bmatrix} 0  &5 & 3\\ 5 & 0 & 1\\ 3 & 1 & 0  \end{bmatrix}$$
It can be shown that for each $i\in\{1,2,3\}$, any allocation of power determined by any power allocation matrix in  $[U_i]_*$, results in  the survival of only country $i$ \cite{survival}. It can also be shown that 
any allocation of power determined by any power allocation matrix in  $[U_4]_*$, causes    all  three countries to be in a   precarious state. All power allocation matrices within this class are symmetric and each one leaves all three countries without any remaining power. These observations motivate the following definition.

 A power allocation matrix $U$ of a power allocation game  is a \emph{balanced equilibrium} if

 \begin{enumerate}
\item $\forall i \in \mathbf{n}$ such that $ \mathcal{A}_{i}$ is the empty set, $u_{ii} = p_{i}$.\label{c1}
\item $\forall i \in \mathbf{n}$ such that $\mathcal{A}_{i} $ is nonempty,
 $u_{ii} = 0$ and $$\sum_{j \in \mathcal{A}_{i}}u_{ij} = p_{i}$$ \label{c2}
\item $\forall (i,j) \in \mathcal{R}_{\mathcal{A}}$, $u_{ij} = u_{ji}$
where $\mathcal{R}_{\mathcal{A}}$  is the set of all pairs $(i,j)$ 
for which countries $i$ and $j$  are adversaries.\label{c3}
\end{enumerate}
It is easy to see that these conditions  imply that  $U$ will be a balanced equilibrium 
if and only if for some permutation $P$, $U = PDP'$  where prime denotes transpose, $D = $block diagonal $\{D_1, D_2\}$,
 $D_1$ is a symmetric matrix with zero diagonals, and $D_2$ is a diagonal matrix whose diagonal entries are the total powers of those countries which have no adversaries.
 Thus in a balanced equilibrium, each country with no adversaries allocates all of its power to itself while each country with adversaries allocates all of its power exclusively to the demise of its 
 adversaries. Each power allocation matrix within the equilibrium class $[U_4]_*$  in Example 1 is a balanced equilibrium.
 We have the following result.

\begin{theorem}
A balanced equilibrium   is a pure strategy Nash equilibrium. 
\label{t1}\end{theorem}

\noindent{\bf Proof of Theorem \ref{t1}:} In a balanced equilibrium $U^{*}$, no country will have incentives to deviate. 

\begin{enumerate}

\item $\forall i \in \mathbf{n}$ s.t. $\mathcal{A}_{i} \neq \emptyset$, $x_{i}(U^{*}) = \text{precarious}$; $\forall j \in \mathcal{A}_{i}$, $x_{j}(U^{*}) = \text{precarious}$.  Therefore, country $i$ has no power to deviate. 

\item $\forall i \in \mathbf{n}$ s.t. $\mathcal{A}_{i} = \emptyset$, $x_{i}(U^{*}) = \text{safe}$; $\forall j \in \mathcal{F}_{i}$, $x_{j}(U^{*}) = \text{safe}$ or $\text{precarious}$ or $\mathcal{F}_{i} = \emptyset$. Therefore, in any of the cases, country $i$ has no incentives to deviate because it has already attained the best power allocation outcome, as implied by the sufficient conditions on the preference relations on the power allocation matrices. 
\end{enumerate} \qed


A balanced equilibrium's realistic implication lies in first illustrating a possible situation multiple parties in conflicts may arrive at.  None has enough power preponderance over others to avoid the precarious state as predicted in the equilibrium. No one will deviate from this equilibrium unless there is a change to their power condition.  

Second, it provides the theoretical basis for a particular kind of military strategy; the Chinese proverb for the military strategy is
``{\em yu bang xiang zheng, yu weng de li}'',
and the English counterpart is ``{\em when shepherds quarrel, the wolf has a winning game} ''.  In other words, by the strategy, the third party can take advantage of others' internal conflicts and thereby achieve its own goals. For example, Catherine the Great hoped to minimize the potential threats to her empire from Austria and Prussia by entangling these two countries in conflicts with France --  as she told her secretary in November 1791, ``I am racking my brains in order to push the courts of Vienna and Berlin into French affairs...there are reasons i cannot talk about; I want to get them involved in that business to have my hands free. I have much unfinished business, and it's necessary for them to be kept busy and out of my way''. \cite{mearsheimer2001tragedy}  It should be noted that the unique survivor respectively corresponding to equilibrium classes 
$[U_1]_*,[U_2]_*,$ and $[U_3]_*$ in Example 1 survives exactly because the other two countries have exhausted themselves in their conflict.

%
%
%
%
%
%
%



It is possible to develop a test for deciding whether or not a given power allocation game has a balanced equilibrium. Towards this end, note first that in the case when no country has any adversaries \{i.e., when all of the $\mathcal{A}_i$ are empty sets\}, the diagonal matrix $U =$ diagonal $\{p_1,p_2,\ldots, p_n\}$ is a balanced equilibrium and it is the only one; this is a consequence of condition \ref{c1} in the definition of a balanced equilibrium. Now suppose that there is at least one adversarial 
relationship between two countries in which case
 there are $n_a <n$ countries with adversaries and $q>0$ adversary pairs in $\mathcal{R}_{\mathcal{A}}$. These    relationships determine  a $q$-edge  subgraph $\mathbb{A}$ of the unsigned version of the  environment graph $\mathbb{E}$ whose $n_a\leq n$ vertices correspond to the $n_a$ countries with adversaries. Let $j_1,j_2,\ldots, j_{n_a}$ denote the labels in $\mathbb{E}$ of the countries which have adversaries and write $\pi$ for the $n_a$-vector $\pi = $ column $\{p_{i_1},p_{i_2},\ldots, p_{i_{n_a}}\}$. Let  $(i_1,j_1), (i_2,j_2), \ldots (i_q,j_q)$ be an ordering of the edges of $\mathbb{A}$, write  $C = [c_{ik}]_{n_a \times q}$ for  the  incidence matrix of $\mathbb{A}$ which is consistent with this ordering, and let $\beta$ denote the map $\beta:\scr{U} \rightarrow \R^{q}$ for which $U\longmapsto $ column $\{u_{i_1j_1},u_{i_2j_2},\ldots,u_{i_qj_q}\}$. 
 Let $\scr{U}_0$ denote the subset of $\scr{U}$
consisting of all power allocation matrices $U$ for which $u_{ii} = p_i$ for all countries $i$ which have no adversaries. 
The definition of a balanced equilibrium implies that  any given balanced equilibrium $U\in\scr{U}$ must be in $\scr{U}_0$ and that any such  $U$  determines   a nonnegative vector $v_U = \beta(U)$  which satisfies
 $Cv_U = \pi$. 
 Conversely any nonnegative vector $v$ satisfying $Cv=\pi$ uniquely determines a balanced equilibrium $U\in\scr{U}_0$ 
 for which $u_{i_kj_k} = v_k = u_{j_ki_k},\;k\in\{1,2,\ldots,q\}.$
  We are led to the following conclusion.

\begin{proposition} Suppose  $\mathcal{R}_{\mathcal{A}}$ is nonempty and  that there are $n_a >1$ countries with adversaries.
Let $\scr{U}_0$,   $C$, $\beta $,  and $\pi$ be as defined above and write  $\scr{V}_{\rm bal}$ for the   set of nonnegative $n_a$-vectors $v$ such that $\pi = Cv$.  
 Then $U$ is a balanced equilibrium in $\scr{U}$ if and only if 
 $U\in\scr{U}_0$ and
 $\beta(U)\in \scr{V}_{\rm bal}$.

\label{t2}\end{proposition}


Deciding whether or not a given matrix $U\in\scr{U}$ is a balanced equilibrium assuming $\scr{R}_{\mathbb{A}}$ is nonempty  thus amounts to checking to see whether or not $U\in\scr{U}_0$ and $\beta(U)\in\scr{V}_{\rm bal}$.
Conversely, any vector $v\in\scr{V}_{\rm bal}$
uniquely determines a balanced equilibrium in $\scr{U}_0$. The problem of finding a vector $v\in\scr{V}_{\rm bal}$ can easily be set up as   a linear programming programming problem. It is clear that  all one needs to do is to minimize the  sum of the $q$  slack variables comprising the $q$-vector $z$ subject to the constraints that $Cv +z = \pi$, $z\geq 0$ and $v\geq 0$. Clearly deciding  whether or not a given power allocation game has 
 a balanced equilibrium can be accomplished  using, for example, the simplex algorithm to determine whether or not the linear programming problem just posed has a solution.

Based on the specific technical structure of balanced equilibrium, the following ``construction lemma'' is presented. Applying it repeatedly, one can start from a simple game that has a balanced equilibrium, and construct a far more complex game while retaining the existence of balanced equilibrium.

\begin{lemma}(\textit{Constructing allocations})
\label{lemma1} Take any $i$, $j \in \mathbf{n}$ of a power allocation game where a balanced equilibrium exists such that $j \in \mathcal{A}_{i}$, and let $\bar{p} = p + \delta(e_{i} + e_{j})$, where $\delta$ is a positive real number. Then the new game also has a balanced equilibrium.  

\end{lemma}
\begin{proof}
Suppose $U$ is a balanced equilibrium of the original power allocation game. Take $(i,j) \in \mathcal{R}_{\mathcal{A}}$, and let $\bar{U} = U + \delta(e_{i}e_{j}^{\top} + e_{j}e_{i}^{\top})$. $\bar{U}$ is a valid power allocation matrix of the new game. 

Obviously, $\bar{U}$ is symmetric. $\forall i ~s.t.~ \mathcal{A}_{i} \neq \emptyset$, $\bar{u}_{ii} = 0$.  No one will deviate, and therefore, it is a balanced equilibrium of the new game. \end{proof}

Lemma~\ref{lemma1} implies that the two games are only different by the power of two countries $i$ and $j$ which have an adversary relation. Lemma~\ref{lemma10} constructs another new node (and its relation) for the original game and a new balanced equilibrium for the new game. 

\begin{lemma}(\textit{Constructing nodes and allocations})
\label{lemma10} Construct a node $n+1 \notin \mathbf{n}$ of a power allocation game where a balanced equilibrium exists. Let $p_{n+1} = 0$, $(i, n+1) \in \mathcal{R}_{\mathcal{A}}$ for one country $i \in \mathbf{n}$. Then the game after incorporating $v_{n+1}$ also has a balanced equilibrium.

\end{lemma}

\begin{proof}
Suppose $U$ is a balanced equilibrium of the original game. Then $\bar{U} = \begin{bmatrix}
U &  0\\
0 &  0
\end{bmatrix}$ is the balanced equilibrium of the new game. \end{proof}

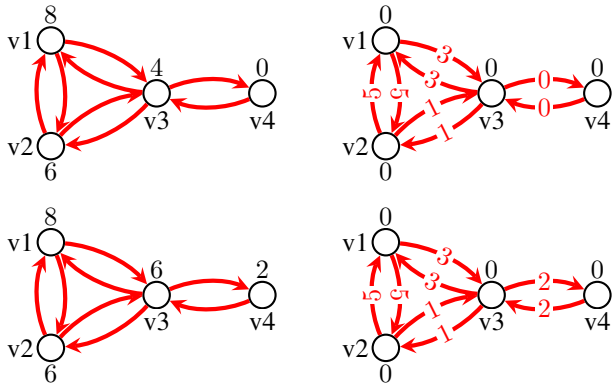
\begin{tikzpicture}[x=4em,y=-4em][scale=0.25]
	\drawnodex{-1,-1/2}{v1}{left}{v1}{above}{$8$}
	\drawnodex{-1,1/2}{v2}{left}{v2}{below}{$6$}
	\drawnodex{0,0}{v3}{below}{v3}{above}{$4$}
	\drawnodex{1,0}{v4}{below}{v4}{above}{$0$}
			
    \drawfoe{v1}{v3}
	\drawfoe{v2}{v3}
	\drawfoe{v1}{v2}
    \drawfoe{v3}{v4}

\end{tikzpicture}
\qquad
\begin{tikzpicture}[x=4em,y=-4em][scale=0.25]
	\drawnodex{-1,-1/2}{v1}{left}{v1}{above}{$0$}
	\drawnodex{-1,1/2}{v2}{left}{v2}{below}{$0$}
	\drawnodex{0,0}{v3}{below}{v3}{above}{$0$}
	\drawnodex{1,0}{v4}{below}{v4}{above}{$0$}
			
    \drawfoex{v1}{v3}{$3$}{$3$}
	\drawfoex{v2}{v3}{$1$}{$1$}
	\drawfoex{v1}{v2}{$5$}{$5$}
    \drawfoex{v3}{v4}{$0$}{$0$}
\end{tikzpicture}

\begin{tikzpicture}[x=4em,y=-4em][scale=0.25]
	\drawnodex{-1,-1/2}{v1}{left}{v1}{above}{$8$}
	\drawnodex{-1,1/2}{v2}{left}{v2}{below}{$6$}
	\drawnodex{0,0}{v3}{below}{v3}{above}{$6$}
	\drawnodex{1,0}{v4}{below}{v4}{above}{$2$}
			
    \drawfoe{v1}{v3}
	\drawfoe{v2}{v3}
	\drawfoe{v1}{v2}
    \drawfoe{v3}{v4}
\end{tikzpicture}
\qquad
\begin{tikzpicture}[x=4em,y=-4em]
	\drawnodex{-1,-1/2}{v1}{left}{v1}{above}{$0$}
	\drawnodex{-1,1/2}{v2}{left}{v2}{below}{$0$}
	\drawnodex{0,0}{v3}{below}{v3}{above}{$0$}
	\drawnodex{1,0}{v4}{below}{v4}{above}{$0$}
			
    \drawfoex{v1}{v3}{$3$}{$3$}
	\drawfoex{v2}{v3}{$1$}{$1$}
	\drawfoex{v1}{v2}{$5$}{$5$}
    \drawfoex{v3}{v4}{$2$}{$2$}
\end{tikzpicture}

The above figure illustrates the process of applying Lemma~\ref{lemma10} to construct a balanced equilibrium for a four-player game:
\begin{enumerate} 
\item Add a country with zero power, country 4, to the original game. 
\item Add $2$ to both $p_{3}$ and $p_{4}$ and obtain the new game. 
\item Construct the balanced equilibriums before and after the change. 
\end{enumerate}

A necessary power condition for countries is now stated for the existence of balanced equilibrium in environments of any relation configuration. The intuition is that if there exists a country whose power does not satisfy the below condition, this country will never exhaust its power by allocating to its adversaries. Theorem 3 (whose proof requires Theorem 2) shows that this is actually the necessary and sufficient condition for countries' power for a balanced equilibrium to exist in a PAG where the adversary pairs make up a complete graph. 


\begin{theorem}(\textit{Necessary power condition})
\label{power} In a power allocation game, the necessary condition for the existence of a balanced equilibrium is that $\forall i \in \mathbf{n}$ s.t. $\mathcal{A}_{i} \neq \emptyset$, $\sum_{j \in \mathcal{A}_{i}}p_{j} \geq p_{i}$.
\end{theorem}
\begin{proof}
By the definition of balanced equilibrium, $\forall i \in \mathbf{n}$ s.t. $\mathcal{A}_{i} \neq \emptyset$, $u_{ii} = 0$ and $u_{ji} = u_{ij}$. And by the power constraint, $\forall j \in \mathcal{A}_{i}$, $u_{ji} \leq p_{j}$. 

Therefore, it is shown that \[
\sum_{j \in \mathcal{A}_{i}}p_{j} \geq \sum_{j \in \mathcal{A}_{i}}u_{ji} = \sum_{j \in \mathcal{A}_{i}}u_{ij} = \sum_{j \in \mathcal{A}_{i}}u_{ij} + u_{ii} = p_{i}
\]
\end{proof}

\begin{lemma}(\textit{Unique Balanced Equilibrium with Three-Player Complete \text{adversary} Graph})
\label{lemma2}
In a power allocation game, there are three countries with \text{adversary} relations. If the \text{adversary} relations make up a complete graph and if the total power condition in theorem \ref{power} holds for them, a unique balanced equilibrium exists.
\end{lemma}
\begin{proof}
A unique balanced equilibrium, $U$, can be constructed. For the three countries country 1, country 2 and country 3 making up a complete graph, their mutual allocations and reserved power are 
\begin{itemize}
\item $[u_{11} \quad u_{12} \quad u_{13}] = [0 \quad \cfrac{p_{1}+p_{2}-p_{3}}{2} \quad \cfrac{p_{1}+p_{3}-p_{2}}{2}]$
\item $[u_{21} \quad u_{22} \quad u_{23}] = [\cfrac{p_{1}+p_{2}-p_{3}}{2} \quad 0 \quad \cfrac{p_{2}+p_{3}-p_{1}}{2}]$
\item $[u_{31} \quad u_{32}\quad u_{33}] = [\cfrac{p_{1}+p_{3}-p_{2}}{2} \quad \cfrac{p_{2}+p_{3}-p_{1}}{2} \quad 0]$

\end{itemize} 
The three invests 0 on the other countries. Those without a \text{adversary} relation invest all their total capacities as their self-defense.  Such a $U$ is a balanced equilibrium, which is also unique.

\end{proof}

\begin{theorem}(\textit{Balanced Equilibrium in Games with Complete \text{adversary} Graph})
\label{prop/complete-balanced equilibrium} In a power allocation game where the \text{adversary} pairs make up a complete subgraph of $\mathbb{E}$, the game has a balanced equilibrium if and only if the necessary total power condition in Theorem~\ref{power} holds for these countries.
\end{theorem}
%
%
%

\noindent \begin{proof} The proof is by induction. 

The Base Case: As proven in Lemma~\ref{lemma2}, a game where the \text{adversary} relations make up a complete 3-country graph has a balanced equilibrium if the total power condition holds for them.

The Induction Hypothesis: a game where the \text{adversary} relations make up a complete k-country graph has a balanced equilibrium if the total power condition holds for them ($k \in \mathbb{Z}$ and $k > 3$).

The Induction Step: Prove that the theorem holds with $k+1$ countries using the assumption above. 

Sort the capacities of the $k+1$ countries by a nonincreasing order: $p_{1} \geq p_{2} ... \geq p_{k+1}$. By the above assumption, the total power condition holds. 


Subtract $p_{1}$ by $p_{k+1}$ and $p_{k+1}$ by $p_{k+1}$. Resort the capacities by a decreasing order: $\bar{p}_{1} \geq \bar{p}_{2} ... \geq \bar{p}_{k} \geq 0$.

The total power condition holds for each of the $k+1$ country. Before the change, $p_{1} \leq p_{2} + p_{3} + ... + p_{k+1}$. After the change, $p_{1} - p_{k+1} \leq p_{2} + p_{3} + ... + p_{k}$. The total power condition still holds for $p_{1}$. 

Since $p_{2} \leq p_{1}$ and $p_{k+1} \leq p_{k}$, it can be proven that $p_{2} \leq (p_{1} - p_{k+1}) + p_{3} + ... + p_{k}$. Then for $p_{2}$, the total power condition still holds. Similarly, the total power condition holds for the rest of the countries. 

By the induction hypothesis, a balanced equilibrium $U$ exists for the game, which has the above $k$ countries with capacities $\bar{p}_{1}$, $\bar{p}_{2}$, ..., and $\bar{p}_{k}$. 

Now another balanced equilibrium, $\bar{U}$ is obtained with the following steps. 
\begin{itemize}

\item Insert a new row after the $k$-th row and a new column after the $k$-th column in $U$, which was originally a $(n-1) \times (n-1)$ matrix after deleting the $(k+1)$-th country, to restore it as a $n \times n$ matrix and to represent the allocations by and towards country $k+1$. 
\item Initialize the elements in the new row and new column as 0. 
\item Then add back the subtracted power $p_{k+1}$ by updating $u_{1,k+1}$ and $u_{k+1,1}$ as $p_{k+1}$. 

\end{itemize}

By the construction lemma, $\bar{U}$ is a balanced equilibrium. Hence, any game where the \text{adversary} relations make up a complete graph has a balanced equilibrium if and only if the total power condition holds.\end{proof}

The below figure illustrates the induction step with a simple case --- at the end of the process, a balanced equilibrium for the original game is constructed.

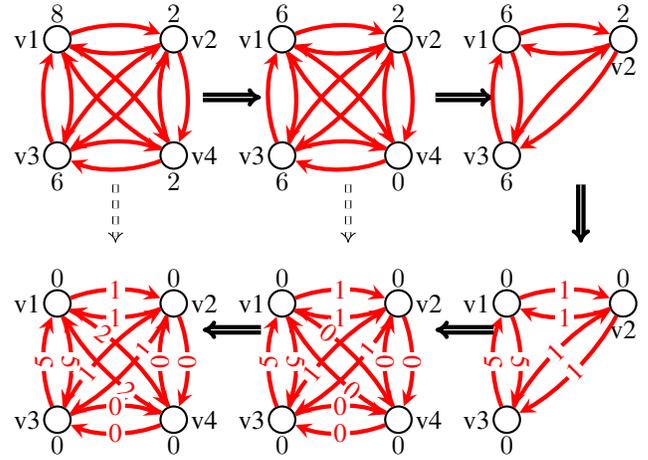
\begin{tikzpicture}[x=2.2em,y=-2.2em]
	\begin{scope}[xshift=0em,yshift=0em]
		\drawnodex{0,0}{v1}{left}{v1}{above}{$8$}
		\drawnodex{2,0}{v2}{right}{v2}{above}{$2$}
		\drawnodex{0,2}{v3}{left}{v3}{below}{$6$}
		\drawnodex{2,2}{v4}{right}{v4}{below}{$2$}
		
		\drawfoe{v1}{v2}
		\drawfoe{v1}{v3}
		\drawfoe{v1}{v4}
		\drawfoe{v2}{v3}
		\drawfoe{v2}{v4}
		\drawfoe{v3}{v4}
	\end{scope}
	\draw[-implies,double,double equal sign distance,ultra thick] (2.5,1) -- +(1,0);
	\draw[-implies,double,double equal sign distance,dashed] (1,2.5) -- +(0,1);
	
	\begin{scope}[xshift=8.5em,yshift=0em]
		\drawnodex{0,0}{v1}{left}{v1}{above}{$6$}
		\drawnodex{2,0}{v2}{right}{v2}{above}{$2$}
		\drawnodex{0,2}{v3}{left}{v3}{below}{$6$}
		\drawnodex{2,2}{v4}{right}{v4}{below}{$0$}
		
		\drawfoe{v1}{v2}
		\drawfoe{v1}{v3}
		\drawfoe{v1}{v4}
		\drawfoe{v2}{v3}
		\drawfoe{v2}{v4}
		\drawfoe{v3}{v4}
	\end{scope}
	\draw[-implies,double,double equal sign distance,ultra thick] (6.5,1) -- +(1,0);
	\draw[-implies,double,double equal sign distance,dashed] (5,2.5) -- +(0,1);
	
	\begin{scope}[xshift=17em,yshift=0em]
		\drawnodex{0,0}{v1}{left}{v1}{above}{$6$}
		\drawnodex{2,0}{v2}{below}{v2}{above}{$2$}
		\drawnodex{0,2}{v3}{left}{v3}{below}{$6$}
		
		\drawfoe{v1}{v2}
		\drawfoe{v1}{v3}
		\drawfoe{v2}{v3}
	\end{scope}
	\draw[-implies,double,double equal sign distance,ultra thick] (9,2.5) -- +(0,1);
	
	\begin{scope}[xshift=17em,yshift=-10em]
		\drawnodex{0,0}{v1}{left}{v1}{above}{$0$}
		\drawnodex{2,0}{v2}{below}{v2}{above}{$0$}
		\drawnodex{0,2}{v3}{left}{v3}{below}{$0$}
		
		\drawfoex{v1}{v2}{1}{1}
		\drawfoex{v1}{v3}{5}{5}
		\drawfoex{v2}{v3}{1}{1}
	\end{scope}
	\draw[implies-,double,double equal sign distance,ultra thick] (6.5,5) -- +(1,0);
	
	\begin{scope}[xshift=8.5em,yshift=-10em]
		\drawnodex{0,0}{v1}{left}{v1}{above}{$0$}
		\drawnodex{2,0}{v2}{right}{v2}{above}{$0$}
		\drawnodex{0,2}{v3}{left}{v3}{below}{$0$}
		\drawnodex{2,2}{v4}{right}{v4}{below}{$0$}
		
		\drawfoex{v1}{v2}{1}{1}
		\drawfoex{v1}{v3}{5}{5}
		\drawfoex[0.3]{v1}{v4}{0}{0}
		\drawfoex[0.3]{v2}{v3}{1}{1}
		\drawfoex{v2}{v4}{0}{0}
		\drawfoex{v3}{v4}{0}{0}
	\end{scope}
	\draw[implies-,double,double equal sign distance,ultra thick] (2.5,5) -- +(1,0);
	
	\begin{scope}[xshift=0em,yshift=-10em]
		\drawnodex{0,0}{v1}{left}{v1}{above}{$0$}
		\drawnodex{2,0}{v2}{right}{v2}{above}{$0$}
		\drawnodex{0,2}{v3}{left}{v3}{below}{$0$}
		\drawnodex{2,2}{v4}{right}{v4}{below}{$0$}
		
		\drawfoex{v1}{v2}{1}{1}
		\drawfoex{v1}{v3}{5}{5}
		\drawfoex[0.3]{v1}{v4}{2}{2}
		\drawfoex[0.3]{v2}{v3}{1}{1}
		\drawfoex{v2}{v4}{0}{0}
		\drawfoex{v3}{v4}{0}{0}
	\end{scope}
\end{tikzpicture}

Theorem 5 discusses an alternative set of countries' power condition for balanced equilibrium to exist in environments where the adversary pairs make up a bipartite graph,  termed as \emph{extended power condition}. This kinds of graphs can also be called \emph{structurally balanced graphs\cite{cartwright1956structural}}.  

Moreover, this extended power condition takes a similar form with the neighbor set condition in the Hall Maximum Matching theorem in Theorem 5. 


\begin{theorem}(\textit{Hall's Maximum Matching Theorem\cite{hall1935representatives}}) Let $\mathbb{G} = (\mathcal{V}, \mathcal{E})$ be a finite bipartite graph where $\mathcal{V} = \mathcal{L} \cup \mathcal{R}$ with $\mathcal{L} \cap \mathcal{R} = \emptyset$. $\mathbb{G}$ has a maximum matching if and only if that for all subsets $\mathcal{S} \in \mathcal{L}$, we have $|\gamma(\mathcal{S})| \geq \mathcal{S}|$ ($\gamma(\mathcal{S}) = \{ j \in \mathcal{R} | (i,j) \in \mathcal{E} ~\text{for some}~ i \in \mathcal{L}\}$), and vice versa. 

\end{theorem}

\begin{theorem}(\textit{Balanced Equilibrium in Games on Structurally Balanced Graph})
\label{prop/bipart-balanced equilibrium} A power allocation game on a structurally balanced graph, in other words, in the environment where the adversary pairs make up a bipartite graph, has a balanced equilibrium if and only if the extended power condition holds for the countries with adversaries:
\begin{itemize}
\item $\forall \mathcal{S} \subseteq \mathcal{L}, \sum\limits_{j \in \mathcal{A}_{\mathcal{S}}}p_{j} \geq \sum\limits_{i \in \mathcal{S}}p_{i}$
\item $\forall \mathcal{S} \subseteq \mathcal{R}, \sum\limits_{j \in \mathcal{A}_{\mathcal{S}}}p_{j} \geq \sum\limits_{i \in \mathcal{S}}p_{i}$ ($\mathcal{A}_{\mathcal{S}} = \bigcup\limits_{i \in \mathcal{S}}\mathcal{A}_{i}$)
\end{itemize}
By definition, the two sets of nodes, $\mathcal{L}$ and $\mathcal{R}$, represent the two groups of countries with \text{adversaries} relations, where each country in either set is only connected to countries in the other set. $\mathcal{S}$ is a subset of either set. 
	
\end{theorem}

\begin{theorem} (\textit{Equivalence between balanced equilibrium on Bipartite \text{adversary} Graph and Max Flow}) \label{prop/flow}For a power allocation game with a bipartite \text{adversary} graph $\mathbb{E} = \{\mathcal{V}, \mathcal{E}\}$ of the left vertex set $\mathcal{L}$ and the right vertex set $\mathcal{R}$, construct a flow network $\bar{\mathbb{E}}$ which satisfies the following 

\begin{itemize}
\item $\bar{\mathcal{V}} = \mathcal{V} \cup \{s, t\}$, with $s$ being the source and $t$ being the sink. 
\item $\bar{\mathcal{E}} = \{(s,i)|i \in \mathcal{L}\} \cup \{(i,j)|i \in \mathcal{L}, j \in \mathcal{R}\} \cup \{(j, t)|j \in \mathcal{R}\}$.
\item Edge capacities\footnote{Edge power in flow networks means the maximum power transmitted on edges.} in $\bar{\mathbb{E}}$ are:  $\bar{c}(s,i)=p_{i}$ if $i \in \mathcal{L}$; $\bar{c}(j,t)=p_{j}$ if $j \in \mathcal{R}$; $\bar{c}(i,j)=+\infty$ if $i \in \mathcal{L}$ and $j \in \mathcal{R}$.
\end{itemize}

The problem of finding a balanced equilibrium in the original game is equivalent to finding a max flow (min cut) $\bar{f}$ in $\bar{\mathbb{E}}$, which satisfies that $\bar{f}(s, i)=\bar{c}(s, i)$ if $i \in \mathcal{L}$ and  $\bar{f}(j,t)=\bar{c}(j,t)$ if $j \in \mathcal{R}$.

\end{theorem}

\section{Conclusions}

An obvious and important extension of this paper is to the case where the antagonism in the networked international environment makes up a $k$-partite graph. Realistically, this can be regarded as a $k$-sided game where countries could be friends with those from the same side and only be adversaries with those from the other $k-1$ sides. Moreover, the $k$-sided game is actually a generalization of a commonly known $k$-player game, which can expect to encompass a series of scenarios in countries' conflicts and cooperation of the real world. 
A natural question to ask is that, in this kind of environment, what kind of power condition for countries is it in order to balanced equilibrium to exist?

\bibliography{balancing}
\bibliographystyle{plain}

\end{document}

%% file: steve.tex
\def\send#1#2{\stackrel{#1}{\hbox to #2{\rightarrowfill}}}
\def\-{\!\!\!\!\!-}

 \def\qed{ \rule{.1in}{.1in}}

\def\scr#1{{\cal #1}}

\def\qed{ \rule{.1in}{.1in}}

\def\R{{\rm I\!R}} 

\newcounter{seqn}[equation]
\def\theseqn{\arabic{equation}\alph{seqn}}

\def\endseqn{\eqno \@seqnnum
$$\ignorespaces}
\def\@seqnnum{(\theseqn)}

\newskip\mcentering \mcentering=0pt plus 1000pt minus 1000pt

\def\meqalignno#1{
\halign to\displaywidth{
    \hbox to 0pt{\kern\displaywidth\llap{$##$}\hss}\tabskip=\mcentering
    &\hfil$\displaystyle{##}$\tabskip=\mcentering
   &&$\displaystyle{{}##}$\hfil\tabskip=\mcentering
    \crcr
    #1\crcr}}

\def\dspace{\multiply\normalbaselineskip 150
		  \divide\normalbaselineskip 100 \normalbaselines
		  \csname @@normalbaselineskip\endcsname\normalbaselineskip}
\def\sspace{\multiply\normalbaselineskip 200
		 \divide\normalbaselineskip 300 \normalbaselines
		 \csname @@normalbaselineskip\endcsname\normalbaselineskip}
\def\sdspace{\multiply\normalbaselineskip 160
		 \divide\normalbaselineskip 150 \normalbaselines
		 \csname @@normalbaselineskip\endcsname\normalbaselineskip}

\def\@{\tilde}

\def\3dot#1{\buildrel\textstyle...\over#1}



%% file: balancing.bbl
\begin{thebibliography}{10}

\bibitem{altafini2012dynamics}
Claudio Altafini.
\newblock Dynamics of opinion forming in structurally balanced social networks.
\newblock {\em PloS one}, 7(6):e38135, 2012.

\bibitem{altafini2013consensus}
Claudio Altafini.
\newblock Consensus problems on networks with antagonistic interactions.
\newblock {\em IEEE Transactions on Automatic Control}, 58(4):935--946, 2013.

\bibitem{altafini2015predictable}
Claudio Altafini and Gabriele Lini.
\newblock Predictable dynamics of opinion forming for networks with
  antagonistic interactions.
\newblock {\em IEEE Transactions on Automatic Control}, 60(2):342--357, 2015.

\bibitem{berge1957two}
Claude Berge.
\newblock Two theorems in graph theory.
\newblock {\em Proceedings of the National Academy of Sciences},
  43(9):842--844, 1957.

\bibitem{cartwright1956structural}
Dorwin Cartwright and Frank Harary.
\newblock Structural balance: a generalization of heider's theory.
\newblock {\em Psychological review}, 63(5):277, 1956.

\bibitem{ford2015flows}
Lester~Randolph Ford~Jr and Delbert~Ray Fulkerson.
\newblock {\em Flows in networks}.
\newblock Princeton university press, 2015.

\bibitem{gerards1995matching}
AMH Gerards.
\newblock Matching.
\newblock {\em Handbooks in operations research and management science},
  7:135--224, 1995.

\bibitem{hall1935representatives}
Philip Hall.
\newblock On representatives of subsets.
\newblock {\em Journal of the London Mathematical Society}, 1(1):26--30, 1935.

\bibitem{hiller2011alliance}
Timo Hiller.
\newblock Alliance formation and coercion in networks.
\newblock 2011.

\bibitem{jackson2005survey}
Matthew~O Jackson.
\newblock A survey of network formation models: stability and efficiency.
\newblock {\em Group Formation in Economics: Networks, Clubs, and Coalitions},
  pages 11--49, 2005.

\bibitem{jackson2015networks}
Matthew~O Jackson and Stephen Nei.
\newblock Networks of military alliances, wars, and international trade.
\newblock {\em Proceedings of the National Academy of Sciences},
  112(50):15277--15284, 2015.

\bibitem{jadbabaie2003coordination}
Ali Jadbabaie, Jie Lin, and A~Stephen Morse.
\newblock Coordination of groups of mobile autonomous agents using nearest
  neighbor rules.
\newblock {\em IEEE Transactions on automatic control}, 48(6):988--1001, 2003.

\bibitem{konig2017networks}
Michael~D K{\"o}nig, Dominic Rohner, Mathias Thoenig, and Fabrizio Zilibotti.
\newblock Networks in conflict: Theory and evidence from the great war of
  africa.
\newblock {\em Econometrica}, 85(4):1093--1132, 2017.

\bibitem{relation}
Yuke Li and A.S. Morse.
\newblock The countries' relation formation problem: I and ii.
\newblock {\em Proceedings of International Federation of Automatic Control
  World Congress}, pages 14141--14146, 2017.

\bibitem{allocation}
Yuke Li and A.S. Morse.
\newblock Game of power allocation on networks.
\newblock {\em Proceedings of American Control Conference}, pages 5231--5236,
  May 2017.

\bibitem{survival}
Yuke Li, A.S. Morse, Ji~Liu, and Tamer Ba\c{s}ar.
\newblock Countries' survival in networked international environments.
\newblock {\em Proceedings of IEEE Conference on Decision and Control}, 2017.

\bibitem{liu2015stability}
Ji~Liu, Xudong Chen, Tamer Ba{\c{s}}ar, and Mohamed~Ali Belabbas.
\newblock Stability of discrete-time altafini's model: a graphical approach.
\newblock In {\em Decision and Control (CDC), 2015 IEEE 54th Annual Conference
  on}, pages 2835--2840. IEEE, 2015.

\bibitem{liu2016discrete}
Ji~Liu, Mahmoud El~Chamie, Tamer Ba{\c{s}}ar, and Beh{\c{c}}et
  A{\c{c}}{\i}kme{\c{s}}e.
\newblock The discrete-time altafini model of opinion dynamics with
  communication delays and quantization.
\newblock In {\em Decision and Control (CDC), 2016 IEEE 55th Conference on},
  pages 3572--3577. IEEE, 2016.

\bibitem{lovasz2009matching}
L{\'a}szl{\'o} Lov{\'a}sz and Michael~D Plummer.
\newblock {\em Matching theory}, volume 367.
\newblock American Mathematical Soc., 2009.

\bibitem{mearsheimer2001tragedy}
John~J Mearsheimer.
\newblock {\em The tragedy of great power politics.}
\newblock WW Norton \& Company, 2001.

\bibitem{meng2016behaviors}
Ziyang Meng, Guodong Shi, Karl~H Johansson, Ming Cao, and Yiguang Hong.
\newblock Behaviors of networks with antagonistic interactions and switching
  topologies.
\newblock {\em Automatica}, 73:110--116, 2016.

\bibitem{olfati2007consensus}
Reza Olfati-Saber, J~Alex Fax, and Richard~M Murray.
\newblock Consensus and cooperation in networked multi-agent systems.
\newblock {\em Proceedings of the IEEE}, 95(1):215--233, 2007.

\bibitem{proskurnikov2014consensus}
Anton Proskurnikov, Alexey Matveev, and Ming Cao.
\newblock Consensus and polarization in altafini's model with bidirectional
  time-varying network topologies.
\newblock In {\em Decision and Control (CDC), 2014 IEEE 53rd Annual Conference
  on}, pages 2112--2117. IEEE, 2014.

\bibitem{proskurnikov2016opinion}
Anton~V Proskurnikov, Alexey~S Matveev, and Ming Cao.
\newblock Opinion dynamics in social networks with hostile camps: Consensus vs.
  polarization.
\newblock {\em IEEE Transactions on Automatic Control}, 61(6):1524--1536, 2016.

\bibitem{shi2016evolution}
Guodong Shi, Alexandre Proutiere, Mikael Johansson, John~S Baras, and Karl~H
  Johansson.
\newblock The evolution of beliefs over signed social networks.
\newblock {\em Operations Research}, 64(3):585--604, 2016.

\bibitem{tsitsiklis1986distributed}
John Tsitsiklis, Dimitri Bertsekas, and Michael Athans.
\newblock Distributed asynchronous deterministic and stochastic gradient
  optimization algorithms.
\newblock {\em IEEE transactions on automatic control}, 31(9):803--812, 1986.

\bibitem{valcher2014consensus}
Maria~Elena Valcher and Pradeep Misra.
\newblock On the consensus and bipartite consensus in high-order multi-agent
  dynamical systems with antagonistic interactions.
\newblock {\em Systems \& Control Letters}, 66:94--103, 2014.

\end{thebibliography}
